\documentclass[11pt, a4paper]{amsart}
\usepackage[latin1]{inputenc}
\usepackage{amssymb, amsmath, amscd, geometry,mathtools}
\usepackage{braket}
\DeclarePairedDelimiter{\abs}{\lvert}{\rvert}
\usepackage{graphicx}

\numberwithin{equation}{section}

\usepackage{latexsym}
\usepackage[usenames]{color}
\usepackage{epic} 
\usepackage{mathrsfs}
\usepackage{euscript}
\usepackage{rotating}
\usepackage{bbm}

\usepackage{verbatim}
\usepackage[usenames]{color}

\theoremstyle{plain}
\newtheorem{theo}{Theorem}[section]

\newtheorem{definition}[theo]{Definition}
\newtheorem{proposition}[theo]{Proposition}
\newtheorem{theorem}[theo]{Theorem}

\newtheorem{lemma}[theo]{Lemma}
\theoremstyle{definition}

\newtheorem{remark}[theo]{\textbf{Remark}}

\newcommand{\sign}{\operatorname{sign}}             

\DeclareMathOperator{\sgn}{sgn}

\begin{document}



\author{}

\title{Quantum one way vs. classical two way communication in XOR games}

\author{Abderram\'an Amr $^1$}
\author{Ignacio Villanueva $^2$}
\address{Departamento de An\'alisis Matem\'atico y Matem\'atica Aplicada, Universidad Complutense de Madrid, Plaza de Ciencias 3, 28040 Madrid, Spain}
\email[$^1$]{abamr@ucm.es}
\email[$^2$]{ignaciov@mat.ucm.es}

\maketitle

\begin{abstract}

In this work we give an example of exponential separation between quantum and classical resources in the setting of XOR games assisted with communication. 

Specifically, we show an example of a XOR game for which $O(n)$ bits of two way classical communication are needed in order to achieve the same value as can be attained with $\log n$ qubits of one way communication.

We also find a characterization for the value of a XOR game assisted with a limited amount of two way communication in terms of tensor norms of normed spaces.\end{abstract}

\section{Introduction and main results}

From the foundations point of view, one of the main goals in Quantum Information is to quantify the difference in performance between quantum and classical resources for a given task. In particular, this quantification has been thoroughly studied in the context of Bell inequalities (see e.g. review \cite{reviewnonlocality}). In the XOR games bipartite scenario, two separate parties, Alice and Bob, are given inputs  $x$ and $y$ and they answer their outputs $a, b=\pm 1$  with certain probability, therefore generating a correlation. The set of correlations they can generate sharing a quantum state and performing local measurements on it is different to the set that they can generate with classical resources, and this difference can be witnessed with the so called Bell inequalities \cite{Bell}. 

Another important setting where this comparison between quantum and classical resources has been studied is the communication complexity scenario \cite{nisan}. In this context, one typically  studies  the minimum number of bits that Alice and Bob have to exchange in order to correctly (up to a bounded probability of error) compute a boolean function for any pair of inputs $x$ and $y$. In this scenario, partial boolean functions have been found for which quantum and classical communication complexities are exponentially separated \cite{expcom}. 

In \cite{JPV} the authors introduced a new setting, which, from a conceptual point of view, can be seen as a mixture of the two previously defined scenarios. In this new setting, Alice and Bob try to win a XOR game using shared classical randomness together with  the communication of a limited number of bits,  either classical or quantum.

In \cite{JPV}, only the one way communication case was studied. That is, the case when  communication is only allowed from one of the parties to the other. Among other results, the authors show an example of a game for which $O(n)$ bits of one way classical communication are needed in order to achieve the same value as the one that can be attained with $\log n$ qubits of one way communication. They left open the question of the existence of a game for which one can obtain the same order of exponential separation between the one way quantum communication and the general (two way) classical communication. 

In our work we answer this question positively. We show that actually the same game appearing in \cite{JPV} achieves this exponential separation, which, as in the one way case, is  the maximum possible separation, up to a logarithmic factor.

We state next some notation needed for the statement of our main result. A more detailed description of our notation can be found in \cite{JPV}. 

A bipartite XOR game $T$ with $R$ and $S$ inputs for Alice and Bob respectively is a linear functional described by a matrix $(T_{x,y})_{x,y=1}^{R,S}$, where $\sum_{x,y} |T_{x,y}|=1$. It describes the situation where Alice and Bob are asked the pair of questions $(x,y)$ with probability $|T_{x,y}|$ and, in order to win the game,  they must output answers $a,b \in \{\pm 1\}$ verifying $ab=\sgn(T_{x,y})$

We call  $\mathcal{L}_{tw,c}$ (respectively $\mathcal{Q}_{ow,c}$) to the convex set of the correlations Alice and Bob can generate when they are allowed the use of  shared randomness and $c$-bits of two way classical communication (respectively $c$ qubits of one way communication). Then, given a XOR game $T$ we can consider the following two quantities: 

\begin{equation*}
\omega_{tw,c}(T)=\sup_{P\in\mathcal{L}_{tw,c}}|\langle T,P\rangle|\text{ and }\omega_{tw,c}^*(M)=\sup_{P\in\mathcal{Q}_{tw,c}}|\langle M,P\rangle|
\end{equation*}

With this notation, our main result can be stated.

\begin{theorem}\label{main}
For every $n\in \mathbb N$, there exist a XOR game $T$ with $2^{2n}$ inputs for Alice and $2^{n^2}$ inputs for Bob such that, for every $k\in \mathbb N$, 
$$\frac{\omega_{ow,\log n}^*(T)}{\omega_{tw,\log k}(T)}\geq C\frac{\sqrt{n}}{\log k},$$
where $C$ is a constant independent of $n,k$.  
\end{theorem}

This result implies the above mentioned exponential separation: Alice and Bob need to communicate $k=O(n)$ classical bits to obtain the same value as the one obtained with $\log n$ qubits.  

The lower bound for the quantum communication value in our result is the same lower bound as in \cite{JPV}. The technical part  of our proof is to upper bound the two way classical communication value. To prove this upper bound we rely on techniques from the local theory of Banach spaces, in particular on a careful use of the Khintchine and double Khintchine inequalities. Also, careful reasoning is needed when handling the dependencies appearing between a message and the previous and following messages. 

Our second result is a characterization of $\omega_{tw, c}(T)$ in terms of tensor norms. Although not strictly needed for Theorem \ref{main}, this second result  was the starting point of this research and lies behind our ideas.

Consider a general two way protocol with $t$-rounds of classical communication. Alice starts the protocol and sends, in the $i$-th round, $c_i$ bits to Bob. After receiving those bits,  Bob sends $d_i$ bits back to Alice. Call momentarily $\omega_{tw}(T)$ to the maximum value that Alice and Bob can obtain with any possible protocol described as above when playing the XOR game $T$. Then our second result characterizes $\omega_{tw}(T)$ in terms of a specific tensor norm on certain spaces. Both the norm and the involved spaces will be defined in Section \ref{s:tensor norm}.

\begin{theorem}\label{norma}
Consider $=(T_{x,y})_{x,y=1}^{R,S}$ and   $\omega_{tw}(T)$  as above. Then the following holds:
$$\omega_{tw}(T)=\|T\otimes id\otimes id\otimes\ldots\otimes id\|_{\ell_1^R(\ell_\infty^{2^{c_1}}(\ell_1^{2^{d_1}}(\ldots(\ell_\infty^{2^{c_t}}(\ell_1^{2^{d_t}})))))\otimes_\epsilon \ell_1^{S 2^{c_1}}(\ell_\infty^{2^{d_1}}(\ldots(\ell_1^{2^{c_t}}(\ell_\infty^{2^{d_t}}))))}.$$
\end{theorem}

Theorem \ref{norma} has a rather clumsy statement due to the intrinsic difficulties of describing two way communication. The idea of the proof is to relate the deterministic communication protocols with the extreme points of the unit ball of $\ell_\infty^{S2^{c_1}}(\ell_1^{2^{d_1}}(\ldots(\ell_\infty^{2^{c_t}}(\ell_1^{2^{d_t}}))))$ and $\ell_\infty^R(\ell_1^{2^{c_1}}(\ell_\infty^{2^{d_1}}(\ldots(\ell_1^{2^{c_t}}(\ell_\infty^{2^{d_t}})))))$, and then to use these extreme points to compute the $\epsilon$ norm \cite{defantfloret} of the operator $T\otimes id\otimes id\otimes\ldots\otimes id$. This norm, in the proper Banach spaces, has already been used to define the classical value of a XOR game \cite{Tsirelson,PVreview} or the value of a XOR game with one way classical communication \cite{JPV}.

As a remark, we mention that  the techniques of Theorem \ref{norma} can be easily  applied to general games, that is, games with a general number of outputs, in order to describe their value when using the above protocol of two way classical communication. 

This paper is organized as follows: In Section \ref{s:cc} we will present the form of a general two way protocol explicitly, with the properties and the dependences of the corresponding messages that are being sent.  In section 3 we will present the proof of Theorem \ref{main}. This proof does not require tensor norms, although, as we said before, it is the tensor norm idea that lies behind our reasonings. Finally, in section 4, we prove Theorem \ref{norma}. In order to do this, we previously state the needed notions from Banach space theory and tensor norm theory. 

\section{Two way classical communication}\label{s:cc}

For the sake of completeness, and in order to fix our notation, we describe next randomized classical communication protocols, and the model associated to them, in the particular case of XOR games. 

We consider a protocol with $t$ rounds of two way classical communication between Alice and Bob. In round  $i$, first Alice will send $c_i$ bits to Bob and, after receiving them, Bob will send $d_i$ bits to Alice. After that, the round $i+1$ can begin.

We consider general randomized protocols and, therefore, the messages each agent sends are  random variables depending on the previous inputs of the corresponding agent. 

That is, we can view the first message $m_1$ of Alice as an application  $$M_1:[R] \longrightarrow \mathbb{R}^{2^{c_1}},$$ such that, for every $x\in [R]$,  $M_1(x):=(M_{1}^{m_1}(x))_{m_1=1}^{2^{c_1}}$ is a probability distribution on the possible messages $m_1$ sent by Alice when she receives input $x$. 

Bob's first message is a mapping $$N_1:[S]\times [2^{c_1}] \longrightarrow  \mathbb{R}^{2^{d_1}},$$ 
such that, for every $y\in [S]$ and $m_1\in [2^{c_1}]$ $N_1(y,m_1):=(N_1^{n_1}(y,m_1))_{n_1=1}^{2^{d_1}}$ is a probability distribution on the possible messages $n_1$ sent by Bob when he receives input $y$ and message $m_1$ from Alice. 

Similarly, Alice's and Bob's last messages are  mappings 
$$M_t:[R]\times [2^{d_1}] \times \dots \times [2^{d_{t-1}}]  \longrightarrow  \mathbb{R}^{2^{c_t}},$$  
and 

$$N_t:[S]\times [2^{c_1}] \times \dots \times [2^{c_{t}}]  \longrightarrow  \mathbb{R}^{2^{d_t}}.$$  

After they interchange messages, Alice and Bob produce ${\pm 1 }$-valued outputs $a(x,n_1, \dots, n_t)$, $b(y, m_1, \dots, m_t)$.

We will use the notation $\overline{m}, \overline{n}$ for the multiindices $(m_1,\dots,m_t)$, $(n_1\dots, n_t)$.  

Therefore, Alice's strategy is a function

 \begin{align*}
&{\bf a}: [R]\times [2^{d_1}]\times \dots \times [2^{d^t}]\longrightarrow \{\pm 1\}\times \mathbb{R}^{2^{c_1}}\times \dots \times \mathbb{R}^{2^{c_t}}\\
&{\bf a}(x,\overline{n})=(a(x,\overline{n}), M_1^{m_1}(x) ,M_2^{m_2}(x,n_1), \dots, M_t^{m_t}(x,n_1, \dots, n_{t-1})),
\end{align*}
which can be seen as a tensor 

\begin{align}\label{e:Alice}
&{\bf \overline{a}}=\left( {\bf \overline{a}}(\tilde{x}, \overline{m}, \overline{n})\right)_{x,  \overline{m},  \overline{n}}=\sum_{x, \overline{m},  \overline{n}} a(x, \overline{n}) M_1^{m_1}(x) M_2^{m_2}(x,n_1) \dots 
\\ \nonumber &\dots M_t^{m_t}(x,n_1, \dots, n_{t-1})
 e_x\otimes e_{m_1}\otimes \dots \otimes e_{m_t}\otimes e_{n_1}\otimes \dots\otimes e_{n_t}.\end{align}

Similarly, Bob's strategy is given by a function 

\begin{align*}
&{\bf b}: [S]\times [2^{c_1}]\times \dots \times [2^{c^t}]\longrightarrow \{\pm 1\}\times \mathbb{R}^{2^{d_1}}\times \dots \times \mathbb{R}^{2^{d_t}}\\
&{\bf b}(y, \overline{m})=(b(y, \overline{m}), N_1^{n_1}(y) ,N_2^{n_2}(y,m_1), \dots, N_t^{n_t}(y,m_1, \dots, m_{t})),
\end{align*}
which can be seen as a tensor 

\begin{align}\label{e:Bob}
&{\bf \overline{b}}= \left( {\bf \overline{b}}(y,  \overline{m}, \overline{n})\right)_{y, \overline{m},  \overline{n} }= \sum_{y, \overline{m}, \overline{n}} b(y, \overline{m}) N_1^{n_1}(y) N_2^{n_2}(y,m_1) \dots \\ &\nonumber \dots N_t^{n_t}(y,m_1, \dots, m_{t}) e_y\otimes e_{m_1}\otimes \dots \otimes e_{m_t}\otimes e_{n_1}\otimes \dots\otimes e_{n_t}.\end{align}

In future reasonings we will need the following result, which follows easily from the definitions. 
\begin{lemma}\label{l:boundednorm}
The tensors ${\bf \overline{a}}, {\bf \overline{b}}$ given in Equations (\ref{e:Alice}) and (\ref{e:Bob}) verify 
  $$\sup_{x} \sum_{m_1}\sup_{n_1}\sum_{m_2}\dots \sup_{n_{t-1}} \sum_{m_t}\sup_{n_t}  |a(x, \overline{n})M_1^{m_1}(x)\dots M_t^{m_t}(x,n_1, \dots, n_{t-1})|\leq 1, $$

and 

$$\sup_{y,m_1} \sum_{n_1}\sup_{m_2}\sum_{n_2}\dots \sup_{m_{t}} \sum_{n_t} |b(y, \overline{m})N_1^{n_1}(y)\dots 
N_t^{n_t}(y,m_1, \dots, m_{t})|\leq 1, $$
\end{lemma}

\begin{proof}
For the first case, bound $a(x, \overline{n})$ by 1 and recall that fixing $x, n_1,\ldots, n_i$ makes $\sum_{m_i}M_i^{m_i}(x,n_1,\ldots,n_{i-1})\leq1$ for all $i$. Proceed similarly for the second case. 
\end{proof}

\section{Proof of Theorem \ref{main}}\label{s:mainproof}

The game appearing in Theorem \ref{main} is the same that was already used in \cite{JPV} to prove a similar bound for the one way communication value. We recall the precise definition of the game here: 

\begin{definition}\label{game}
We consider the XOR game $T$ where the input of Alice is an element $\tilde{x}=(x,z) \in \{\pm 1\}^n \times \{\pm 1\}^n$ and the input of Bob is an element  $y\in\{\pm 1\}^{n^2}$. Then the coefficients $T_{\tilde{x},y}=T_{(x,z),y}$ take the following form: 
$$T_{(x,z),y}=\frac{1}{L}\sum_{i,j=1}^n x_iz_j y_{ij}$$

Where $L$ is a normalization factor in order to fulfill $\sum_{xzy}|T_{(x,z),y}|=1$, which means $L=\sum_{xyz}|\sum_{ij}x_iz_j y_{ij}|$.

\end{definition}

That is, the probability of question $(\tilde{x}, y)$ is  $\frac{1}{L}|\sum_{ij}x_iz_j y_{ij}|$ and the condition that the players have to fulfill with their answers in that case is $ab=\sign \sum_{i,j}x_i z_j y_{ij}$.

\begin{remark}\label{estimate}
The following estimate for the value of $L$ is given in \cite[Lemma 5.3]{JPV}:

$$\frac{1}{\sqrt{2}}n2^{n^2+2n}\leq L\leq n2^{n^2+2n}$$
\end{remark}

In order to prove Theorem \ref{main} we need to show a lower bound for the value with quantum communication and an upper bound for the value with classical communication. The quantum value was already proven in \cite{JPV}.

\begin{proposition}\cite[Proposition 5.6]{JPV}\label{lower bound XOR game}
Let $T$ be the XOR game defined in Definition \ref{game}. Then, 
\begin{align*}
\omega^*_{ow,\log n}(T)\geq  \frac{C}{\sqrt{n}},
\end{align*}where  $C$ is a constant independent of $n$.
\end{proposition}

Our main contribution is the upper bound for the value with two way classical communication.  To make the proof easier to follow, we state first some lemmas. Some of them were already used in \cite{JPV}, but we recall them here for completeness and the convenience of the reader. 

First we state Khintchine and Double Khintchine inequalities in the precise form we will use. A proof of the double Khintchine inequality can be found in \cite[pag. 455]{defantfloret}. 

\begin{theorem}[Khintchine inequalities]\label{l:Khintchine}
For $1\leq p <\infty$ there exist constants $a_p, \text{}b_p\geq 1$ such that
\begin{align}\label{standard Khintchine ineq}
a_p^{-1}\left(\sum_{i=1}^n |\alpha_i|^2\right)^\frac{1}{2}\leq \left(\sum_{y\in \{\pm 1\}^n}\frac{1}{2^n}\Big|\sum_{i=1}^n \alpha_i y_i\Big|^p\right)^{\frac{1}{p}} \leq b_p\left(\sum_{i=1}^n |\alpha_i|^2\right)^\frac{1}{2}
\end{align}for every $n\in \mathbb N$ and all $\alpha_1, \cdots, \alpha_n \in \mathbb C$.

Moreover, 
\begin{align}\label{double Khintchine ineq}
a_p^{-2}\left(\sum_{i,j=1}^n |\alpha_{i,j}|^2\right)^\frac{1}{2}\leq \left(\sum_{x,z\in \{\pm 1\}^n}\frac{1}{2^{2n}}\Big|\sum_{i,j=1}^n \alpha_{i,j}x_i z_j \Big|^p\right)^{\frac{1}{p}} \leq b_p^2\left(\sum_{i,j=1}^n |\alpha_{i,j}|^2\right)^\frac{1}{2}
\end{align}for every $n\in \mathbb N$ and all $\alpha_{1,1}, \alpha_{1,2}, \cdots, \alpha_{n,n} \in \mathbb C$.
\end{theorem}

In our reasonings we actually need the trasposed version of both Khintchine inequalities. We state the precise result. 

\begin{lemma}\label{l:transpose Khintchine} 
Let $1<p<\infty$ and let $p'$ be such that $\frac{1}{p}+\frac{1}{p'}=1$. 
Then, for every $n\in \mathbb N$ and for  every  sequence of numbers $(\alpha(y))_{y\in \{-1, 1\}^n}$,
\begin{align*}
\left(\sum_{i=1}^n \left( \sum_{y\in \{-1, 1\}^n} y_i \alpha(y)\right)^2\right)^\frac{1}{2}\leq b_{p'}^2  \left(2^{n}\right)^\frac{1}{p'} \left(\sum_{y\in \{-1, 1\}^n} |\alpha(y)|^p\right)^\frac{1}{p},
\end{align*}
where  $b_{p'}$ is the constant appearing in Lemma \ref{l:Khintchine} for $p'$.

Moreover, for every $n\in \mathbb N$ and for every finite sequence of numbers $(\alpha(x,z))_{(x,z)\in \{-1, 1\}^n\times  \{\pm 1\}^n}$,
\begin{align*}
\left(\sum_{i,j=1}^n \left( \sum_{(x,z)} x_i z_j \alpha(x,z)\right)^2\right)^\frac{1}{2}\leq b_{p'}^2  \left(2^{2n}\right)^\frac{1}{p'} \left(\sum_{(x,z)} |\alpha(x,z)|^p\right)^\frac{1}{p},
\end{align*} where  the sums in $(x,z)$ are over $\{\pm 1\}^n\times  \{\pm 1\}^n$ and $b_{p'}$ is  again the constant appearing in Lemma \ref{l:Khintchine}  for $p'$.
\end{lemma}

\begin{proof}
The second statement follows from (\ref{double Khintchine ineq}). The proof can be seen in \cite[Lemma 5.4]{JPV}. The proof of the first statement is done similarly, using (\ref{standard Khintchine ineq}) rather than  (\ref{double Khintchine ineq}).
\end{proof}

We will also need the following simple consequence of  Holder's inequality. 

\begin{lemma}\label{lemma1}
For every $1<p<\infty$ and for every finite sequence of real numbers $(\alpha_i)_{i=1}^d$, 
$$\sum_{i=1}^d \abs{\alpha_i}\leq d^{1/p'}\Big(\sum_{i=1}^d\abs{\alpha_i}^p\Big)^{1/p},$$
where $\frac 1p + \frac{1}{p'}=1$
\end{lemma}

We state and prove  one more technical simple result.

\begin{lemma}\label{messagesbound}
Let ${\bf \overline{a}}, {\bf \overline{b}}$ be as in Equations (\ref{e:Alice}), (\ref{e:Bob}). Then, for every $(x, y)\in \mathcal X\times \mathcal Y$
$$\sum_{\overline{m}, \overline{n}}| {\bf \overline{a}}(x,\overline{m},\overline{n}) {\bf \overline{b}}(y, \overline{m}, \overline{n})|\leq1$$
\end{lemma}

\begin{proof}
Recalling the definitions of ${\bf \overline{a}}, {\bf \overline{b}}$, we have
\begin{align*} 
&\sum_{\overline{m}, \overline{n}}| {\bf \overline{a}}(x,\overline{m},\overline{n}) {\bf \overline{b}}(y, \overline{m}, \overline{n})| \\
&=\sum_{\overline{m}, \overline{n}} |   a(x,\overline{n}) M_1^{m_1}(x) M_2^{m_2}(x,n_1) \dots M_t^{m_t}(x,n_1, \dots, n_{t-1})   b(y,\overline{m}) N_1^{n_1}(y) N_2^{n_2}(y,m_1) \dots\\
& \dots N_t^{n_t}(y,m_1, \dots, m_{t})|\\
&\leq\sum_{\overline{m}, n_1, \dots, n_{t-1}} |  b(y,\overline{m})|  M_1^{m_1}(x) M_2^{m_2}(x,n_1) \dots M_t^{m_t}(x,n_1, \dots, n_{t-1}) \\
& N_1^{n_1}(y) N_2^{n_2}(y,m_1) \dots  N_{t-1}^{n_{t-1}}(y,m_1, \dots, m_{t-1}) \sum_{n_t} |a(x, \overline{n})|N_t^{n_t}(y,m_1, \dots, m_{t}) \\
&\leq \sum_{\overline{m}, n_1, \dots, n_{t-1}}|  b(y, \overline{m})|  M_1^{m_1}(x) M_2^{m_2}(x,n_1) \dots M_t^{m_t}(x,n_1, \dots, n_{t-1}) \\
& N_1^{n_1}(y) N_2^{n_2}(y,m_1) \dots  N_{t-1}^{n_{t-1}}(y,m_1, \dots, m_{t-1})\\
&=\sum_{\substack{m_1,\dots, m_{t-1}\\n_1,\dots, n_{t-1}}}   M_1^{m_1}(x) M_2^{m_2}(x,n_1) \dots M_t^{m_t}(x,n_1, \dots, n_{t-1}) \\
& N_1^{n_1}(y) N_2^{n_2}(y,m_1) \dots  N_{t-1}^{n_{t-1}}(y,m_1, \dots, m_{t-1})  \sum_{m_t} |  b(y,\overline{m})|M_t^{m_t}(x,n_1, \dots, n_{t-1}) \\
&\leq \sum_{\substack{m_1,\dots, m_{t-1}\\n_1,\dots, n_{t-1}}}   M_1^{m_1}(x) M_2^{m_2}(x,n_1) \dots M_t^{m_t}(x,n_1, \dots, n_{t-1}) \\
& N_1^{n_1}(y) N_2^{n_2}(y,m_1) \dots  N_{t-1}^{n_{t-1}}(y,m_1, \dots, m_{t-1})\leq 1.
\end{align*}
To see  the last inequality, it is enough to keep on summing in the same order, that is, in $n_{t-1}$, then in $m_{t-1}$, then in $n_{t-2}$, etc. 
\end{proof}

Now we can upper bound the value of $T$ with two way classical communication. We have 

\begin{proposition}\label{p:twbound}

Let $T$ be the XOR game from Definition \ref{game}. Then 

$$\omega_{tw, \log k} (T)\leq \frac{4\sqrt{2} e^{5/2}(\log k)^{3/2}}{n}.$$
\end{proposition}

\begin{proof} We assume there are $t$ rounds of communication with a total amount of bits exchanged of $\log k$. Therefore, $\log k=\sum_{i=1}^t c_i+d_i$, where $c_i, d_i$ are as in Section \ref{s:cc}. We also assume that Alice starts the communication, the other case being similar. 

As explained in Section \ref{s:cc}, it is enough to bound the quantity 

$$\sum_{\substack{\tilde{x}, y \\ \overline{m}, \overline{n}}}  T_{\tilde{x},y} {\bf \overline{a}}(\tilde{x},\overline{m},\overline{n}) {\bf \overline{b}}(y,\overline{m},\overline{n}),$$
when $ {\bf \overline{a}}(\tilde{x},\overline{m},\overline{n})$, $ {\bf \overline{b}}(y,\overline{m},\overline{n})$ are as in Equations (\ref{e:Alice}) and (\ref{e:Bob}).

We have 

\begin{align*}
&\sum_{\substack{\tilde{x}, y \\ \overline{m}, \overline{n}}}  T_{\tilde{x},y} {\bf \overline{a}}(\tilde{x},\overline{m},\overline{n}) {\bf \overline{b}}(y,\overline{m},\overline{n})\leq \sum_{ \overline{m}, \overline{n}} \left|\sum_{\tilde{x}, y } T_{\tilde{x},y} {\bf \overline{a}}(\tilde{x},\overline{m},\overline{n}) {\bf \overline{b}}(y,\overline{m},\overline{n})\right| \\
&\leq k^{\frac{1}{p'}}  \left(\sum_{ \overline{m}, \overline{n}} \left|\sum_{\tilde{x}, y } T_{\tilde{x},y} {\bf \overline{a}}(\tilde{x},\overline{m},\overline{n}) {\bf \overline{b}}(y,\overline{m},\overline{n})\right|^p\right)^\frac{1}{p}\\
&=\frac{k^{\frac{1}{p'}}}{L}  \left(\sum_{ \overline{m}, \overline{n}} \left|\sum_{(x,z), y } \sum_{i,j} x_i z_j y_{ij}  {\bf \overline{a}}(\tilde{x},\overline{m},\overline{n}) {\bf \overline{b}}(y,\overline{m},\overline{n})\right|^p\right)^\frac{1}{p}\\
&=\frac{k^{\frac{1}{p'}}}{L}  \left(\sum_{ \overline{m}, \overline{n}} \left|\sum_{i,j}\left( \sum_{(x,z) }  x_i z_j {\bf \overline{a}}(\tilde{x},\overline{m},\overline{n}) \right)\left(\sum_y y_{ij}   {\bf \overline{b}}(y,\overline{m},\overline{n})\right)\right|^p\right)^\frac{1}{p},
\end{align*}
where the second inequality follows from Lemma \ref{lemma1}. 

We note now that, for every choice of $\overline{m}, \overline{n}$,

\begin{align*}
& \left|\sum_{i,j}\left( \sum_{(x,z) }  x_i z_j {\bf \overline{a}}(x,z,\overline{m},\overline{n}) \right)\left(\sum_y y_{ij}   {\bf \overline{b}}(y,\overline{m},\overline{n})\right)\right| \\
&\leq \left(\sum_{i,j}\left( \sum_{(x,z) }  x_i z_j {\bf \overline{a}}(x,z,\overline{m},\overline{n}) \right)^2\right)^\frac{1}{2} \left(\sum_{i,j} \left(\sum_y y_{ij}   {\bf \overline{b}}(y,\overline{m},\overline{n})\right)^2\right)^\frac{1}{2} \\
&\leq b_{p'}^3 \left( 2^{2n+n^2}\right)^\frac{1}{p'}  \left( \sum_{x,z }  \left| {\bf \overline{a}}(x,z,\overline{m},\overline{n}) \right |^p\right)^\frac{1}{p} \left( \sum_{y }  \left| {\bf \overline{b}}(y,\overline{m},\overline{n}) \right |^p\right)^\frac{1}{p}, 
\end{align*}
where the first inequality follows from Cauchy-Schwartz inequality and the second one follows from Lemma \ref{l:transpose Khintchine}. 

Using this, we have that 
\begin{align*}
&\sum_{\substack{\tilde{x}, y \\ \overline{m}, \overline{n}}}  T_{\tilde{x},y} {\bf \overline{a}}(\tilde{x},\overline{m},\overline{n}) {\bf \overline{b}}(y,\overline{m},\overline{n}) \\
&\leq \frac{k^{\frac{1}{p'}}}{L}  b_{p'}^3 \left( 2^{2n+n^2}\right)^\frac{1}{p'} \left( \sum_{\overline{m}, \overline{n} } \left( \sum_{x,z }  \left| {\bf \overline{a}}(x,z,\overline{m},\overline{n}) \right |^p\right)  \left( \sum_{y }  \left| {\bf \overline{b}}(y, \overline{m},\overline{n}) \right |^p\right) \right)^\frac{1}{p}\\
&= \frac{k^{\frac{1}{p'}}}{L}  b_{p'}^3 \left( 2^{2n+n^2}\right)^\frac{1}{p'} \left(  \sum_{x,z,y } \sum_{\overline{m}, \overline{n} }   \left|    {\bf \overline{a}}(x,z,\overline{m},\overline{n})  {\bf \overline{b}}(y, \overline{m},\overline{n})\right  |^p \right)^\frac{1}{p} \\
& \leq  \frac{k^{\frac{1}{p'}}}{L}  b_{p'}^3 \left( 2^{2n+n^2}\right)^\frac{1}{p'}  \left( 2^{2n+n^2}\right)^\frac{1}{p},
\end{align*}
where in the last inequality we have used Lemma \ref{messagesbound} and the simple fact that, for every $1<p<\infty$,  if $$\sum_{\overline{m}, \overline{n} }   \left|    {\bf \overline{a}}(x,z,\overline{m},\overline{n})  {\bf \overline{b}}(y, \overline{m},\overline{n})\right  | \leq 1,$$
then also 
$$\sum_{\overline{m}, \overline{n} }   \left|    {\bf \overline{a}}(x,z,\overline{m},\overline{n})  {\bf \overline{b}}(y, \overline{m},\overline{n})\right  |^p \leq 1.$$

To finish, we use that $L\geq \frac{1}{\sqrt{2}}n2^{n^2+2n}$ by Remark \ref{estimate}. We also use that  $b_{p'}\leq \sqrt{2e p'}$ (see \cite[Section 8.5]{defantfloret}) and we make the choice $p'=\log k$. Then we have:

\begin{align*}
&\sum_{\substack{\tilde{x}, y \\ \overline{m}, \overline{n}}}  T_{\tilde{x},y} {\bf \overline{a}}(\tilde{x},\overline{m},\overline{n}) {\bf \overline{b}}(y,\overline{m},\overline{n})\leq \frac{4 e^\frac{5}{2} (\log k)^\frac{3}{2}}{n}.
\end{align*}

\end{proof}

Now, Propositions \label{p:twbound} and {lower bound XOR game} together prove Theorem \ref{main}.

\section[Tensor norm for two way communication]{The value of a game with two way classical communication as a tensor norm}\label{s:tensor norm}

The purpose of this section is to show that the value of any XOR game assisted with a general two way classical communication protocol can be described by a norm in the tensor of  certain Banach spaces. In order to make this work self contained, the required notions and definitions from Banach space theory and tensor norm theory will be presented here. 

Given a normed space $X$, denote by $\|\cdot\|_X$ its norm, and by $B_X=\{x\in X\text{ such that }\|x\|_X\leq1\}$ its unit ball. The dual space consists of the linear and continuous maps from $X$ to the scalar field ($\mathbb{R}$ in our case) and it is denoted by $X^*$. The norm of the dual space has the natural expression $\|x^*\|_{X^*}=\sup_{x\in B_X}\abs{\braket{x^*,x}}$.

All Banach spaces considered in this article are finite dimensional. In particular we are interested in the spaces $\ell_1^R$ and $\ell_\infty^R$, and their combination which we describe below.  

Given a Banach space $X$, we will define the spaces  $\ell_1^R(X)$ and $\ell_\infty^R(X)$: As vector spaces, they are just the spaces whose elements are sequences of $R$ elements in $X$. Given one such element  $u=\{x_i\}_{i=1}^R$ with $x_i\in X$, their norms are defined as follows:
\begin{align*}
&\|u\|_{\ell_1^R(X)}=\sum_{i=1}^R\|x_i\|_X, \\
&\|u\|_{\ell_\infty^R(X)}=\max_{1\leq i \leq R} \|x_i\|_X.
\end{align*} 

With this definition at hand, we will consider the spaces $\ell_1^R(\ell_\infty^S)$,  $\ell_\infty^R(\ell_1^S)$ and further concatenation of these spaces. For example, the element 
\begin{equation}\label{Eq:z} z=\{z(x_1,a_1,x_2,a_2,\ldots,x_t,a_t)\}_{x_1,a_1,\ldots,x_t,a_t}\in \mathbb{R}^{R_1S_1\ldots R_t S_t}\end{equation}
can be seen as an element in  the space $\ell_\infty^{R_1}(\ell_1^{S_1}(\ldots\ell_\infty^{R_t}(\ell_1^{S_t})\ldots))$. Considered in that space, the norm of $z$ is 
$$\|z\|_{\ell_\infty^{R_1}(\ell_1^{S_1}(\ldots\ell_\infty^{R_t}(\ell_1^{S_t})\ldots))}=\max_{x_1}\sum_{a_1}\ldots\max_{x_t}\sum_{a_t} \abs{z(x_1,a_1,\ldots,x_t,a_t)}.$$
Note the similarity of this expression with the one appearing in Lemma \ref{l:boundednorm}.

Recall that a sequence of $R$ elements in $X$, $u=\{x_i\}_{i=1}^R$ can be naturally   seen as an element in the tensor product $\mathbb{R}^R\otimes X$, the identification being  $u=\sum_{i=1}^R e_i\otimes x_i$, where $e_i$ are the vectors of the canonical basis of $\mathbb R^R$. Hence, the element $z$ mentioned in (\ref{Eq:z}) can be naturally identified with an element in  $\mathbb{R}^{R_1}\otimes\mathbb{R}^{S_1}\otimes\ldots\otimes\mathbb{R}^{R_t}\otimes\mathbb{R}^{S_t}$.

Given two finite dimensional Banach spaces $X$ and $Y$, the tensor product  $X\otimes Y$ can be endowed with different norms compatible with the norm structure of $X$ and $Y$, giving raise to different Banach spaces. This is the core idea of tensor norm theory. In this work, we will need the so called $\epsilon$-norm.  The following definition of the $\epsilon$-norm,  toghether with basic properties thereof,  can be seen, for instance,  in \cite{defantfloret,ryan}. 

Given two normed spaces $X, Y$ and an element  $u=\sum_{i=1}^L x_i\otimes y_i$ in  $X\otimes Y$, the $\epsilon$-norm of $u$ is defined by:
 
\begin{align}\label{Def_eps_pi}
\|u\|_{X\otimes_\epsilon Y}&=\sup\Big\{\abs{\sum_{i=1}^L |x^*(x_i)| | y^*(y_i)|}:\,  x^*\in B_{X^*},y^*\in B_{Y^*}	\Big\}.
\end{align}

We will use the notation $X\otimes_\epsilon Y$ to refer to the space $X\otimes Y$ endowed with the $\epsilon$-norm.

Some basic notions about convexity will also be needed. Recall that a set $A$ is convex if given $x$ and $y$ in $A$, then $\lambda x+(1-\lambda)y$ is in $A$ for all $\lambda$ in $[0,1]$. Given a set with $n$ elements $B=\{x_1,\ldots,x_n\}$, we define the convex hull of $B$ as: 
$$co(B)=\{\sum_{i=1}^{n}\alpha_i x_i\text{ such that }x_i\in ,\alpha_i\geq0,\sum_{i=1}^{n} \alpha_i=1\}$$

An extreme point of a set $A$ is a point which does not lie in any open line segment joining two points in the set. That is, if $y$ is an extreme point of $A$ and we can write $y=\lambda x_1+(1-\lambda) x_2$ with $x_1$ and $x_2$ in $A$, and with $x_1\neq x_2$, then $\lambda$ is either 0 or 1. It is well known and easy to see that every convex set coincides with  the convex hull of its extreme points. 

The proof of the following two lemmas follows immediately from the definitions involved. 

\begin{lemma}\label{extremal1}
Denoting by $\{e_i\}_{i=1}^R$ to the elements of the canonical basis of $\mathbb R^R$, we have: 
\begin{enumerate}
\item The extreme points of $B_{\ell_\infty^R}$ are exactly the elements of the form $\sum_{i=1}^R a_i e_i$, where  $a_i=\pm 1$ for every $i$.  

\item  The extreme points of $B_{\ell_1^R}$ are exactly the elements of the form $a_i e_i$, where $a_i=\pm 1$.

\item  Given a Banach space $X$, the extreme points of $B_{\ell_\infty^R (X)}$ are exactly the elements of the form $\sum_{i=1}^R e_i \otimes x_i$, where $x_i$ is an extreme point of $B_X$ for every $i$, and we use the tensor notation to identify  $\ell_\infty^R(X)$ and $\ell_\infty^R \otimes X$

\item Given a Banach space $X$, the extreme points of $B_{\ell_1^R(X)}$ are exactly the elements of the form $e_i \otimes x_i$, where $x_i$ is an extreme point of $B_X$ and we use the tensor notation as above.
\end{enumerate}

\end{lemma}

In the reasonings below, it will be useful to write $e_i \otimes x_i$ as $\sum_{j=1}^R
\delta_{i,j} e_j\otimes x_i $.

The following result characterizes the extreme points of the unit ball of the space $\ell_\infty^{R_1}(\ell_1^{S_1}(\ldots\ell_\infty^{R_t}(\ell_1^{S_t})\ldots))$.

\begin{lemma}\label{extremal} The extreme points of the unit ball of  $\ell_\infty^R (\ell_1^{2^{c_1}}(\ell_\infty^{2^{d_1}}(\ldots(\ell_1^{2^{c_t}}(\ell_\infty^{2^{d_t}}))\ldots)))$  are exactly the elements of  the form: 

$$\sum_{x,\overline{m}, \overline{n}}z_{x,\overline{n}}\delta_{m_1,m_1(x)}\delta_{m_2,m_2(x,n_1)}\ldots\delta_{m_t,m_t(x,n_1,\ldots,n_{t-1})}e_x\otimes e_{m_1}\otimes e_{n_1}\otimes\ldots\otimes e_{m_t}\otimes e_{n_t},$$
where $z_{x,n_1,\ldots,n_t}=\pm1$ for all $x$, $n_1,\ldots,n_t$ and $m_1:[R]\rightarrow[2^c_1]$, $m_2:[R]\times[2^{d_1}]\rightarrow[2^{c_2}]$ and so on, are functions.

Similarly, the extreme points of the unit ball of $\ell_\infty^{S 2^{c_1}}(\ell_1^{2^{d_1}}(\ell_\infty^{2^{c_2}}(\ell_1^{2^{d_2}}(\ldots(\ell_\infty^{2^{c_t}}(\ell_1^{2^{d_t}}))\ldots))))$ are exactly the elements with the form: 

$\sum_{y,n_1,m_1,\ldots, n_t,m_t}z_{y,m_1,\ldots,m_t}\delta_{n_1,n_1(y,m_1)}\delta_{n_2,n_2(y,m_1,m_2)}\ldots\delta_{n_t,n_t(y,m_1,m_2,\ldots,m_t)}e_y\otimes e_{m_1}\otimes e_{n_1}\otimes\ldots\otimes e_{m_t}\otimes e_{n_t}$, 
where, similarly as above,  $z_{x,m_1,\ldots,m_t}=\pm1$ for all $x$, $m_1,\ldots,m_t$ and $n_1:[S]\rightarrow[2^{d_1}]$, $n_2:[S]\times[2^{c_1}]\rightarrow[2^{d_2}]$ and so on, are functions. 

\end{lemma}

\begin{proof}
The proof follows easily from Lemma \ref{extremal1} and induction. For the sake of clarity we write out the proof for the case of  $\ell_\infty^{S 2^{c_1}}(\ell_1^{2^{d_1}}(\ell_\infty^{2^{c_2}}(\ell_1^{2^{d_2}})))$, which corresponds to $t=2$ in the second statement of the Lemma.

First note that following Lemma \ref{extremal1} and the notation following it,  the extreme elements of the unit ball of $\ell_\infty^{2^{c_2}}(\ell_1^{2^{d_2}})$ are of the form  

$$\sum_{m_2, n_2=1}^{2^{c_2}, 2^{d_2}} z_{m_2}\delta_{n_2,n_2(m_2)}e_{m_2}\otimes e_{n_2},$$

where $n_2: [2^{c_2}]\rightarrow [2^{d_2}]$ runs over all possible functions  and $z_{m_2}=\pm1$ for all $m_2$. 

Then, with the aid of the $\delta$ notation, the extreme points of the of the unit ball of $\ell_1^{2^{d_1}}(\ell_\infty^{2^{c_2}}(\ell_1^{2^{d_2}}))$ can be written as 

\begin{align}\label{exex}\sum_{n_1=1}^{2^{d_1}}\sum_{n_2,m_2}z_{m_2}\delta_{n_2,n_2(m_2)}\delta_{n_1,n_0}e_{n_1}\otimes e_{m_2}\otimes e_{n_2},\end{align} 
where $n_0\in [2^{d_1}]$.

Finally, to describe the extreme points of the unit ball of $\ell_\infty^{S 2^{c_1}}(\ell_1^{2^{d_1}}(\ell_\infty^{2^{c_2}}(\ell_1^{2^{d_2}})))$, first note that $\mathbb R^{S 2^{c_1}} = \mathbb R^{S}\otimes \mathbb R^{ 2^{c_1}}$. Then, applying again Lemma \ref{extremal1}, for every $y$ and $m_1$, we obtain that the extreme points of the unit ball of  $\ell_\infty^{S 2^{c_1}}(\ell_1^{2^{d_1}}(\ell_\infty^{2^{c_2}}(\ell_1^{2^{d_2}})))$ are exactly those of the form

$$\sum_{y, m_1=1}^{S, 2^{c_1}}e_y\otimes e_{m_1} \otimes \left(\sum_{n_1=1}^{2^{d_1}}\sum_{n_2,m_2}z_{m_2}\delta_{n_2,n_2(m_2)}\delta_{n_1,n_0}e_{n_1}\otimes e_{m_2}\otimes e_{n_2}\right). $$

In that expression, the functions $z_{m_2}$, $\delta_{n_2,n_2(m_2)}$ and $\delta_{n_1,n_0}$  depend also on $y$ and $m_1$, and therefore we can rewrite the formula above as:
$$ \sum_{y, m_1,n_1,m_2,n_2} z_{y,m_1, m_2} \delta_{n_1,n_1(m_1,y)}\delta_{n_2,n_2(y,m_1,m_2)} e_y\otimes e_{m_1}\otimes e_{n_1}\otimes e_{m_2}\otimes e_{n_2},$$
where  $n_2:[S]\times[2^{c_1}]\times[2^{c_2}]\rightarrow[2^{d_2}]$ and $n_1:[S]\times[2^{c_1}]\rightarrow[2^{d_1}]$ are functions  and $z_{y,m_1,m_2}=\pm1$ for all $y,m_1,m_2$. \end{proof}

Using the lemmas stated above, we can now find an expression for the value of a XOR game $T=(T_{xy})_{x,y=1}$ with the protocol defined in Section \ref{s:cc}, in which there is a total amount of $c$-bits of two way communication exchanged, in $t$ different rounds. The messages sent by Alice to Bob use $c_1$ to $c_t$ bits respectively, and the ones sent by Bob to Alice, $d_1$ to $d_t$, respectively. Hence $\sum_{i=1}^t c_i + \sum_{i=1}^t d_i=c$. 
\begin{proof}[Proof of Theorem \ref{norma}]
Considering the supremum below in the possible strategies of Alice and Bob, we have 

\begin{align*}
\omega_{tw,c}(T)=&\sup \sum_{x,\overline{m}, \overline{n}} T_{x,y} M_1^{m_1}(x)N_1^{n_1}(y,m_1)\dots \\ 
 & \dots M_t^{m_t}(x,n_1,\dots, n_{t-1}) N_t^{n_t}(y, m_1,\dots, m_{t}) a(x, \overline{n}) b(y,\overline{m})=\\
=&\sup \sum_{x,y,\overline{m}, \overline{n}, \overline{m}', \overline{n}'} \delta_{m_1,m'_1}\dots \delta_{m_t,m'_t} \delta_{n_1,n'_1}\dots \delta_{n_t,n'_t}T_{x,y} M_1^{m_1}(x)N_1^{n'_1}(x,m'_1)\dots \\ 
& \dots M_t^{m_t}(x,n_1,\dots, n_{t-1}) N_t^{n'_t}(y,m'_1,\dots,m'_t) a(x, \overline{n}) b(y,m'_1, \dots, m'_t)=\\
=& \langle T\otimes id\otimes\dots\otimes id | \sum_{\substack{x,m_1,\dots, m_t,\\ n_1\dots, n_t}}a(x, \overline{n})M_1^{m_1}(x)\dots \\
&M_t^{m_t}(x,n_1,\dots, n_{t-1}) e_x\otimes e_{m_1}\otimes \dots \otimes e_{m_t} \otimes e_{n_1}\otimes \dots \otimes e_{n_t}\otimes \\
&\sum_{\substack{y, m'_1,\dots, m'_t,\\ n'_1\dots, n'_t}}  b(y,m'_1, \dots, m'_t) N_1^{n'_1}(x,m'_1)\dots N_t^{n'_t}(y,m'_1,\dots,m'_t) \\
& e_y\otimes e_{m'_1}\otimes \dots \otimes e_{m'_t} \otimes e_{n'_1}\otimes \dots \otimes e_{n'_t} \rangle\\
\end{align*}

We recommend the reader to write the formula above in the case $t=2$.

It follows now immediately from the  definitions and Lemma \ref{l:boundednorm} that $$\omega_{tw,c}(T)\leq \|T\otimes id\otimes\dots\otimes id\|_{\ell_1^R(\ell_\infty^{2^{c_1}}(\ell_1^{2^{d_1}}(\ldots(\ell_\infty^{2^{c_t}}(\ell_1^{2^{d_t}})))))\otimes_\epsilon\ell_1^{S 2^{c_1}}(\ell_\infty^{2^{d_1}}(\ldots(\ell_1^{2^{c_t}}(\ell_\infty^{2^{d_t}}))))}$$

In order to prove  the reverse inequality, note first that it follows from the definitions that $\|T\|_{\ell_1^R(\ell_\infty^{2^{c_1}}(\ell_1^{2^{d_1}}(\ldots(\ell_\infty^{2^{c_t}}(\ell_1^{2^{d_t}})))))\otimes_\epsilon\ell_1^{S 2^{c_1}}(\ell_\infty^{2^{d_1}}(\ldots(\ell_1^{2^{c_t}}(\ell_\infty^{2^{d_t}}))))}$ coincides with  
$$\sup\left\{\left\langle T \otimes id\otimes\dots\otimes id| {\bf a}\otimes {\bf b} \right\rangle \right\}$$ when ${\bf a}\in B_{\ell_\infty^R(\ell_1^{2^{c_1}}(\ell_\infty^{2^{d_1}}(\ldots(\ell_1^{2^{c_t}}(\ell_\infty^{2^{d_t}})))))}$ and ${\bf b}\in B_{\ell_\infty^{S 2^{c_1}}(\ell_1^{2^{d_1}}(\ldots(\ell_\infty^{2^{c_t}}(\ell_1^{2^{d_t}}))))}$. It follows now from compactness and convexity that the supremum above is actually a maximum which will be attained on extreme points ${\bf a, b}$ of the respective unit balls.

Now, Remark \ref{extremal} tells us that the extreme points of $B_{\ell_\infty^R(\ell_1^{2^{c_1}}(\ell_\infty^{2^{d_1}}(\ldots(\ell_1^{2^{c_t}}(\ell_\infty^{2^{d_t}})))))}$ have the form 

$$\sum_{x,n_1,m_1,\ldots ,n_t,m_t}z_{x,n_1,\ldots,n_t}\delta_{m_1,m_1(x)}\delta_{m_2,m_2(x,n_1)}\ldots\delta_{m_t,m_t(x,n_1,\ldots,n_{t-1})}e_x\otimes e_{\gamma_1}\otimes e_{\delta_1}\otimes\ldots\otimes e_{\gamma_t}\otimes e_{\delta_t}$$

which can be seen according to (\ref{e:Alice}) as a deterministic strategy for Alice in which the final answer is $z_{x,n_1,\ldots,n_t}$ and the messages that she has send are $m_1(x)$, $m_2(x,n_1)$,... and $m_t(x,n_1,\ldots,n_{t-1})$. 

We proceed similarly for the extreme points of $B_{\ell_\infty^{S 2^{c_1}}(\ell_1^{2^{d_1}}(\ldots(\ell_\infty^{2^{c_t}}(\ell_1^{2^{d_t}}))))}$. 

\end{proof}

\section*{acknowledgment}
This research was funded by the Spanish MINECO through Grant No. MTM2017-88385-P and MTM2014-54240-P.

\end{document}